\documentclass[letterpaper, 12pt,journal,onecolumn,singleside]{ieeetran}\usepackage{amsmath,amssymb,euscript ,yfonts,psfrag,latexsym,dsfont,graphicx,bbm,color,amstext,wasysym,subfig,parskip,pdfsync,epstopdf,balance}
\usepackage{amsmath ,amssymb,euscript ,yfonts,psfrag,latexsym,dsfont,graphicx,bbm,color,amstext,wasysym,subfig,flushend,parskip,url}

\usepackage{flushend}
\usepackage{amsthm}
\graphicspath{{./},{./figures/},{./figures/drawings/},{../figures/}}

\newcounter{plm_example}
\newcounter{plm_thm}
\newtheorem{thm}[plm_thm]{Theorem}

\newtheorem{example}[plm_example]{Example}

\definecolor{grey}{rgb}{0.7,0.7,0.7}
\definecolor{lgrey}{rgb}{0.9,.7,0.7}

\newcommand{\zz}{{\zeta}}

\begin{document}
\def\spacingset#1{\def\baselinestretch{#1}\small\normalsize}

\setlength{\parindent}{10pt}
\parskip 2pt
\def\spacingset#1{\def\baselinestretch{#1}\small\normalsize}
\newcommand{\E}{\operatorname{E}}
\newcommand{\bX}{\mathbf X}
\newcommand{\bH}{\mathbf H}
\newcommand{\bA}{\mathbf A}
\newcommand{\bB}{\mathbf B}
\newcommand{\bN}{\mathbf N}

\newcommand{\bHo}{{\stackrel{\circ}{\mathbf H}}}

\newcommand{\bHom}{{\stackrel{\circ\,\,}{{\mathbf H}_{t}}}{\hspace*{-13pt}\phantom{H}}^-}
\newcommand{\bHop}{{\stackrel{\circ\,\,}{{\mathbf H}_{t}}}{\hspace*{-13pt}\phantom{H}}^+}
\newcommand{\bHotwo}{{\stackrel{\circ\,\,\,\,\,}{{\mathbf H}_{t_2}}}{\hspace*{-16pt}\phantom{H}}^-}
\newcommand{\bHomn}{{\stackrel{\circ}{{\mathbf H}}}{\hspace*{-9pt}\phantom{H}}^-}
\newcommand{\bHopn}{{\stackrel{\circ}{{\mathbf H}}}{\hspace*{-9pt}\phantom{H}}^+}

\newcommand{\EbHom}{\E^{{\stackrel{\circ\,\,}{{\mathbf H}_{t}}}{\hspace*{-11pt}\phantom{H}}^-}}
\newcommand{\EbHop}{\E^{{\stackrel{\circ\,\,}{{\mathbf H}_{t}}}{\hspace*{-10pt}\phantom{H}}^+}}

\spacingset{1.2}

\title{Dynamic relations in sampled processes
\thanks{Supported by AFOSR-FA9550-17-1-0435, 
NSF-ECCS-1509387, 
ARO-W911NF-17-1-0429, 
 and SSF.}
}

\author{Tryphon T. Georgiou 
and Anders Lindquist
\thanks{T.T.\ Georgiou is with the Department of Mechanical and Aerospace Engineering,
University of California, Irvine, California; {email: tryphon@uci.edu}}
\thanks{A.\ Lindquist is with the Department of  Automation and the School of Mathematics, Shanghai Jiao Tong University, Shanghai, China, and the Department of Mathematics at 
KTH Royal Institute of Technology, Stockholm, Sweden; {email: alq@kth.se}}
} 

\pagestyle{empty}
\maketitle
\thispagestyle{empty}
\begin{abstract}
Linear dynamical relations that may exist in continuous-time, or at some natural sampling rate, are not directly discernable at reduced observational sampling rates.
Indeed, at reduced rates, matricial spectral densities of vectorial time series have maximal rank and thereby cannot be used to ascertain potential dynamic relations between their entries.
This hitherto undeclared source of inaccuracies appears to plague off-the-shelf identification techniques seeking remedy in hypothetical observational noise.
In this paper we explain the exact relation between stochastic models at different sampling rates and show how to construct
stochastic models at the finest time scale that data allows. 
We then point out that the correct number of dynamical dependencies can only be ascertained by considering stochastic models at this finest time scale, which in general is faster than the observational sampling rate.
%
\end{abstract}

\newcommand{\mR}{{\mathbb R}}
\newcommand{\mZ}{{\mathbb Z}}
\newcommand{\mN}{{\mathbb N}}
\newcommand{\mE}{{\mathbb E}}
\newcommand{\mC}{{\mathbb C}}
\newcommand{\mD}{{\mathbb D}}
\newcommand{\bU}{{\mathbf U}}
\newcommand{\bW}{{\mathbf W}}
\newcommand{\cF}{{\mathcal F}}

\newcommand{\trace}{{\rm trace}}
\newcommand{\rank}{{\rm rank}}
\newcommand{\Real}{{\Re}e\,}
\newcommand{\half}{{\frac12}}

\section{Introduction}
\noindent
\IEEEPARstart{S}{uppose} that we seek to identify linear dynamical relations that may exist between the components of a continuous-time process. What is typically available to us is the discrete-time sampled process (time series) of measurements collected at a given finite sampling rate. In this paper we are not concerned with issues of statistical estimation but instead assume that we can determine sufficiently accurately the spectral density of the sampled process.
The theme of this paper is on how, from the model parameters of the sampled process, one may obtain a maximal number of dependencies between the entries of the process at a suitably finer time scale.

Modern-day applications, which aim towards high dimensional data and possibly varying sampling protocols, further underscore the importance of a careful consideration of how sampling affects 
dependencies. For example, consider a network as in Figure~\ref{fig:network} with a continuous stationary stochastic vector process
$\zeta=(\zeta_1,\zeta_2,\dots,\zeta_p)'$, 
whose components correspond to specific nodes. Suppose there is a deterministic dynamical relationship, described by the transfer function $T(s)$, between $u=(\zeta_1,\zeta_2,\dots,\zeta_m)'$ and $y=(\zeta_{m+1},\zeta_{m+2},\dots,\zeta_p)'$ as depicted in Figure~\ref{fig:fig1}. 
An important question is whether it is possible to detect this dynamical dependency from an observed sampled process $\zeta_k :=\zeta(kh)$, where $h$ is the sampling period. 
The same issue occurs when the original  process is discrete-time and the observed process is subsampled at a slower rate. 

\begin{figure}[htb]\begin{center}
\hspace*{-15pt}\includegraphics[width=0.32\textwidth]{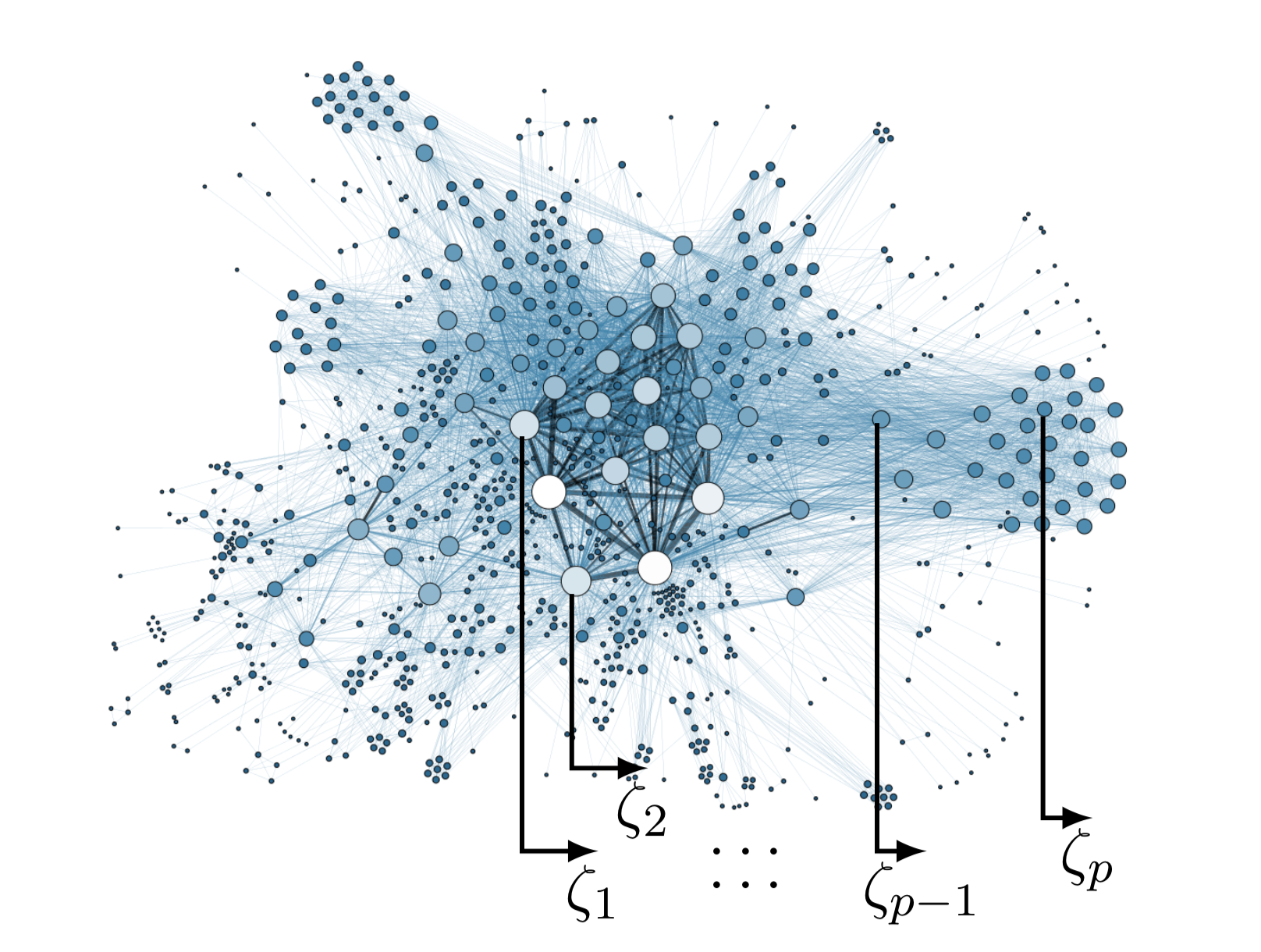}
   \caption{Processes in a network environment}
   \label{fig:network}
\end{center}\end{figure}
The preponderance of techniques in the literature implicitly assume that the observed time series inherits any dynamical dependencies and typically seek to identify relations for the discrete-time process at the observation sampling rate. However this assumption cannot be made in general. Dependencies that may exist in continuous-time, or at some other fine ``natural'' sampling rate, are obfuscated by the process of sampling or sub-sampling. This is reflected in the fact that while the nullity of the spectral density of the original process coincides with the number of linear relations, the density of the sampled process has generically maximal rank. Evidently, when the observational sampling rate is sufficiently fast, the density of the sampled process is close to a singular one with the correct nullity. In such cases,
\begin{figure}[htb]
\begin{center}
\hspace*{-10pt}\includegraphics[width=0.38\textwidth]{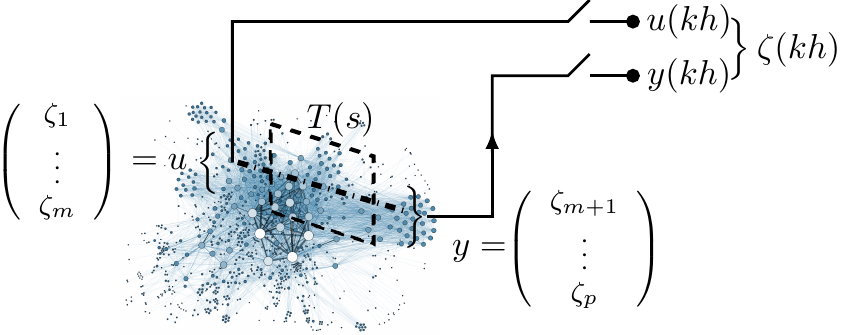}
   \caption{Dynamical (network) dependencies}
   \label{fig:fig1}
\end{center}\end{figure}
hypothesizing a ``small amount'' of observational noise may salvage the situation and suggest a nearly correct model for dynamic dependencies.
But in general there are no guarantees that hypothesizing observational noise, especially when there is none, as exemplified in Figure \ref{fig:fig1}, will not exacerbate statistical inaccuracies.

The stochastic nature of the two vector processes $u$ and $y$ in  Figure \ref{fig:fig1} originates in the ``environment'' of uncertain components and noise sources in feedback, while no observation noise is hypothesized. The setting represents a rather typical situation in distributed systems where variables can be accurately measured, albeit at a finite rate.  The purpose of this work is to point out that the statistics/model of the sampled or subsampled process contain enough information to correctly identify a ``primitive'' model for the underlying finer-time-scale process at its {\em natural finer time-resolution}. Such a primitive model, by necessity, exposes any dynamical relations between variables that may exist. This observation suggests the following potentially enabling new {\em dictum for linear system identification}: given time series data that has been collected at a sampling period $h$, and a derived model at that time resolution, consider the family of models at sampling periods $h/q$, with $q$ an integer, or in continuous time, for underlying processes that are consistent with the observations at the particular sampling period $h$. Amongst those models select the one that exposes a maximal number of dynamical relationships. 
This will be the model with the finest time-resolution, and it will be referred to as the {\em lifting\/} of the observation-time-scale model to the natural time-resolution. 

The outline of the paper is as follows. In Section \ref{sec:sec2} we explain how dynamical dependencies between the entries of a vector process are encoded in the spectral density and how they can be recovered via spectral factorization in the form of a transfer matrix that relates respective components. In Section \ref{sec:c2d} we provide the correspondence between continuous-time stochastic systems and their corresponding sampled discrete-time versions. In particular, Theorem \ref{prop:discretemodel} provides necessary and sufficient conditions for a discrete-time model to originate via sampling. Then, Section \ref{sec:sec5} details the correspondence between sampled models at different sampling rates. In each of these two sections we present an academic example that illustrates  key points.  In Section \ref{sec:sec4} we investigate general discrete-time linear stochastic model as to when they can be lifted to continuous-time ones. Finally, Section \ref{sec:conclusions} concludes with implications of the theory.

\section{Dynamic relations and spectral rank}\label{sec:sec2}

We consider finite-dimensional stationary Gauss-Markov real stochastic processes and are interested in identifying the maximal 
number of deterministic dynamical relations between the entries of such processes from observational data. Initially, in this section, we
explain {\em how such relations relate to the spectral density of the process at its natural time resolution}, whether this is in continuous-time as depicted in Figure \ref{fig:fig1} or discrete-time at some basic sampling rate.

To exemplify the task of identifying dynamical dependencies consider the setting of Figure \ref{fig:fig1}
and let
\begin{equation}\label{eq:z}
\zz(t):= \left(\begin{matrix}u(t)\\y(t)\end{matrix}\right)
\end{equation}
be a stationary process partitioned into two components.
Its stochastic nature originates in the ``environment'' that may include multiple feedback loops and noise sources. 
Assume that a deterministic dynamical relation between  $u$ and $y$, represented by the transfer function $T$, exists. This transfer function does not need to be stable -- only the combined (feedback) dynamics need to be. Moreover, causation is not essential in our discussion, although we will see shortly how the entries of $\zz$ can be partitioned into ``inputs'' and ``outputs'' so as to assign a consistent {\em proper} transfer function that models their dependence.

Continuing in the context of Figure \ref{fig:fig1}, the spectral density
\begin{align}\nonumber
\Phi_\zz(i\omega)&= \left(\begin{matrix} \Phi_u(i\omega) & \Phi_{uy}(i\omega)\\\Phi_{yu}(i\omega) &\Phi_y(i\omega)\end{matrix}\right),
\end{align}
 of the stochastic process $\zz$ factors as
\begin{align}\label{eq:outer}
\Phi_\zz(i\omega)&=\left(\begin{matrix} I\\ T(i\omega)\end{matrix}\right) \Phi_u(i\omega)\left(\begin{matrix} I,& T(-i\omega)'\end{matrix}\right).
\end{align}
Thus, assuming that $\Phi_u$ is  a.e.\ nonsingular, the number of deterministic dynamical relations between the entries of $\zz$ coincides with the nullity of $\Phi_\zz$. Indeed, 
the transfer function between $u$ and $y$ can be recovered from $\Phi_\zz$ since
\begin{equation}\label{eq:T}
T(s)=\Phi_{yu}(s)\Phi_u(s)^{-1},
\end{equation} 
and the deterministic dependence between the entries of $\zz$ can be expressed in ``kernel'' form as
$L(i\omega) \Phi_\zz(i\omega) =0$ with $L(s) := (-T(s),\; I)$.

We now explain how such dynamical relations can be readily obtained in state-space form via spectral factorization of $\Phi_\zz$.
Throughout, we let $\zz$ be a stationary $p$-vector stochastic process with the $p\times p$ rational spectral density $\Phi_\zz(i\omega)$ having
\[
\rank(\Phi_\zz(i\omega))=m, \mbox{ a.e., for }\omega\in\mathbb R.
\]
Thus,
$
m\leq p,
$
and when the inequality is strict, there are $p-m$ dynamical relations between the entries of $\zz$. For our purposes we assume throughout that $\Phi_\zz(\infty)=0$. This ensures that $\zz$ has continuous sample paths, which is needed later on when we define the sampled process $\zz_k:=\zz(kh)$. 
Denote by
\[
V_\zz(s)=H(sI-F)^{-1}G
\]
a minimal stable $p\times m$ spectral factor of $\Phi_\zz$  (see e.g., \cite[page 198]{LPbook}). That is,
\begin{equation}\label{eq:factorization}
\Phi_\zz(i\omega)=V_\zz(i\omega)V_\zz(-i\omega)^\prime
\end{equation}
with $H\in \mR^{p\times n}$, $F\in \mR^{n\times n}$, $G\in \mR^{n\times m}$, $(F,G)$ reachable, $(F,H)$ observable, and $F$ stable matrix, i.e., having all eigenvalues in the open left halfplane.
In particular,  $G$ has full column rank and 
\[
\rank(G)=m\leq n. 
\]
Moreover, we assume that $\rank(HG)=m$. 
Below we provide a factorization of $V_\zz$ that displays the transfer function $T$;
for notational convenience, rational proper functions with possibly a constant term will be displayed in the following standard matrix-notation
\[
C(sI-A)^{-1}B+D=:\left[\begin{array}{c|c}A&B\\\hline C&D\end{array}\right].
\]

\begin{thm}
With $V_\zz,F,G,H$ as above, re-order the rows of $H$ and partition
\[H=\left(\begin{matrix}H_0\\H_1\end{matrix}\right)
\]
so that $H_0G$ is $m\times m$ and invertible. Re-order in the same way the entries of $\zz$ and as in \eqref{eq:z} let $u$ represent the first $m$ entries and $y$ the remaining. Then,
\begin{equation}
V_\zz(s)= \left(\begin{matrix} I\\T(s)\end{matrix}\right)M(s)s^{-1}, \mbox{ where}
\end{equation}
\begin{subequations}
\begin{align}\nonumber\\[-.05in]
\label{eq:T(s)}
T(s)&=\left[\begin{array}{c|c}\Gamma &G (H_0G)^{-1}\\\hline H_1\Gamma &H_1G(H_0G)^{-1}\end{array}\right],\\
M(s)&=\left[\begin{array}{c|c}F&G\\\hline H_0F&H_0G\end{array}\right], \mbox{ and}\\
\Gamma&=F-G(H_0G)^{-1}H_0F.
\end{align}
\end{subequations}
\end{thm}

\begin{proof}
Partition $V_\zz$ conformably with $H$, to write
\begin{align*}
V_\zz(s)
&=\left(\begin{matrix}s^{-1}(H_0G+H_0F(sI-F)^{-1}G)\\
s^{-1}(H_1G+H_1F(sI-F)^{-1}G)\end{matrix}\right)\\
&=:\left(\begin{matrix}M(s)\\N(s)\end{matrix}\right)s^{-1}.
\end{align*}
The only thing that needs to be shown is the claimed expression in \eqref{eq:T(s)}. One can verify directly that
$T(s)M(s)$
coincides with $N(s)=H_1G+H_1F(sI-F)^{-1}G$. To carry out the computations and verify, one needs to write the product in state-space form and cancel the modes corresponding to the system matrix $\Gamma$.
\end{proof}

A corresponding result can be easily worked out in discrete time where the spectral factor may also have a constant term.

\section{Continuous-time to discrete-time and back}\label{sec:c2d}

Once again we consider $\zz$ to be a stationary $p$-vector process with $p\times p$ rational spectral density $\Phi_\zz(i\omega)$ of constant rank $m$ a.e.\ on $i\mathbb R$,
and $F,G,H$ as specified in the previous section.
A minimal Markovian representation  \cite{LPbook} of $\zz$ is
\begin{subequations}\label{eq:continuous}
\begin{align}
\label{eq:continuous_model}
dx(t)&=Fx(t)dt + Gdw(t)\\
\label{eq:dx}
\zz(t)&=Hx(t)
\end{align}
\end{subequations}
with $w(t)$ a standard vectorial Wiener process of compatible dimension.
The sampled paths of $x(t)$ are also continuous a.s.\ and
the corresponding sampled process $x_k=x(kh)$
satisfies
\[
x_{k+1}=e^{Fh}x_k+\int_{kh}^{(k+1)h}e^{F((k+1)h-\tau)}Gdw(\tau).
\]
It follows that $x_k$ and $\zz_k=\zz(kh)$ satisfy the discrete-time stochastic system of equations
\begin{subequations}\label{eq:discrete}
\begin{align}\label{eq:xk}
x_{k+1}&= A x_k +Bv_k\\\label{eq:zk}
\zz_k&=Cx_k
\end{align}
\end{subequations}
with,
\begin{subequations}\label{eq:parameters}
\begin{align}
A&=e^{Fh} \label{A}\\
C&=H\label{C}\\
Bv_k&=\int_{kh}^{(k+1)h}e^{F((k+1)h-\tau)}Gdw(\tau),
\end{align}
and $v_k$ a sequence of independent random vectors having zero mean and unit variance (i.e., normalized white noise). 
Then
\begin{align*}
BB'&=\int_0^h e^{F\tau}GG'e^{F'\tau}d\tau =:Q.
\end{align*}

Since $(F,G)$ is a controllable pair, $Q=BB'$ is nonsingular
and, in particular, we can take
\begin{align}\label{eq:BQhalf}
B&=Q^{1/2}.
\end{align}
\end{subequations}
At the same time,
\begin{align}\nonumber
BB'\hspace*{-2pt}=\hspace*{-2pt}&\int_0^\infty \hspace*{-6pt}e^{F\tau}GG'e^{F'\tau}d\tau- \underbrace{e^{Fh}}_{A}
\left(\int_0^\infty  \hspace*{-6pt}e^{F\tau}GG'e^{F'\tau}d\tau \hspace*{-2pt}\right) \underbrace{e^{F'h}}_{A'},\nonumber
\end{align}
and therefore,
\[
P:=\int_0^\infty e^{F\tau}GG'e^{F'\tau}d\tau,
\]
satisfies simultaneously two Lyapunov equations
\begin{subequations}
\begin{align}
\label{eq:discreteLyapunov}
P&=APA'+BB', \mbox{ and}
\end{align}
\begin{align}\label{eq:contLyap}
FP+PF'+GG'&=0,
\end{align}
\end{subequations}
i.e., both a ``discrete-time'' and a ``continuous-time'' Lyapunov equation.

We distill the bijection between the continuous-time model and its sampled counterpart in the following theorem.\\[-0.05in]

\begin{thm}\label{prop:discretemodel}
Consider the continuous-time stochastic model \eqref{eq:continuous} having parameters $(F,G,H)$ with $(F,G)$ controllable. 
Sampling with period $h$ gives rise to the discrete-time stochastic model \eqref{eq:discrete} with parameters $(A,B,C)$ satisfying (\ref{eq:parameters}). Conversely, if the parameters $(A,B,C)$ of 
stochastic model \eqref{eq:discrete} (with $A$ a stability matrix) are such that
\begin{subequations}\label{conditions}
\begin{align}
\mbox{i) }&A \mbox{ admits a (principal) matrix logarithm \cite{culver1966existence}}
\label{eq:logm}\\
& \mbox{that we denote }\log(A),\nonumber\\
\mbox{ii) }&\det (BB')\neq 0, \label{Bsqinv}\\
\label{eq:key}
\mbox{iii) }&\log(A)P+P\log(A)'\leq 0,\mbox{ for }P\\
&\mbox{ the solution of } P=APA'+BB',\nonumber
\end{align}
\end{subequations}
then \eqref{eq:discrete} arises by sampling \eqref{eq:continuous} with $F=\frac{1}{h}\log(A)$, $H=C$,
and $G$ is a left factor of $-(FP+PF')$ of full column rank.
\end{thm}

\begin{proof} The first claim, to the effect that sampling leads to (\ref{eq:parameters}), has already been shown. For the second part, given any sampling period $h$, let $F=\frac{1}{h}\log(A)$, $H=C$, and
$G$ such that
\begin{equation}\label{eq:LyapforGG}
GG'=-(FP+PF'),
\end{equation}
it can be readily verified that conditions \eqref{eq:parameters} hold.
That \eqref{eq:LyapforGG} can be solved for a suitable $G$ follows from \eqref{eq:key}. 
\end{proof}

The transfer function corresponding to \eqref{eq:discrete} is
\[
W(z)=C(zI-A)^{-1}B,
\]
and the spectral density of the sampled process $\zz_k$ is
\begin{align}
\Psi(z)&= W(z)W(z^{-1})'\nonumber\\
&=\label{eq:Psi}
C(zI-A)^{-1}Q(z^{-1}I-A')^{-1}C'.
\end{align}
Therefore, since $Q>0$ and $\rank ((zI-A)^{-1})=n$, we have that $\rank (\Psi(z)) =\rank(C)$ a.e. On the other hand, $m=\rank(\Phi_\zz(s))\leq \rank(H)$. Hence, since $\rank(C)=\rank(H)\geq m$, 
\[
p\geq\rank(\Psi(z))\geq \rank(\Phi_\zz(s))= m.
\] 
Thus the rank of $\Psi$ does not immediately reveal the $p-m$ linear dependencies. However $m$ can be computed as the rank of $\log(A)P+P\log(A)'$ in the second part of Theorem \ref{prop:discretemodel}.


Finally, 
since $A$ is invertible (a necessity, since $A$ admits a principal matrix logarithm \cite{culver1966existence}), $W(0)=-CA^{-1}B$ is finite, and since $W(\infty)=0$,
\begin{equation}\label{eq:Psi0}
\Psi(0)=0.
\end{equation}
These two conditions, $A$ admitting a matrix logarithm and \eqref{eq:Psi0}, characterize spectral densities of sampled processes.

\begin{example}\label{ex:example2}
{\em   We now present an illustrative academic example.
To this end we consider the continuous-time model \eqref{eq:continuous} with 
\begin{equation}\label{eq:FFtotal}
F=\left(\begin{matrix}\phantom{-}0&\phantom{-}1&\phantom{-}0&\phantom{-}0\\
\phantom{-}0&\phantom{-}0&\phantom{-}1&\phantom{-}0\\
\phantom{-}0&\phantom{-}0&\phantom{-}0&\phantom{-}1\\
-2&-5&-6&-4\end{matrix}\right), G=\left(\begin{matrix}1&-1\\1&-1\\1&-1\\1&\phantom{-}1\end{matrix}\right), \\\nonumber
\end{equation}
\begin{equation}\label{eq:HH}
H=\left(\begin{matrix}1&2&3&4\\
2&1&2&3\\
3&2&1&2\\
\vdots&\ddots&\ddots&\\
9&8&7&6\\
10&9&8&7
\end{matrix}\right).
%
\end{equation}
The rank of the corresponding spectral density is $2$, in agreement with the rank of $G$.
Next, we consider stochastic models $(A,B,C)$ for the sampled processes with sample periods $h\in \{1, 2^{-1}, 2^{-2},\ldots, 2^{-7}\}$. Their spectral density has rank $4$, and so does $BB'$. Thus, in 
Figure \ref{fig:fig2} we display the four eigenvalues of $BB'$  in logarithmic scale as a function of $h$.
\begin{figure}\begin{center}
\hspace*{15pt}
    \includegraphics[width=0.38\textwidth]{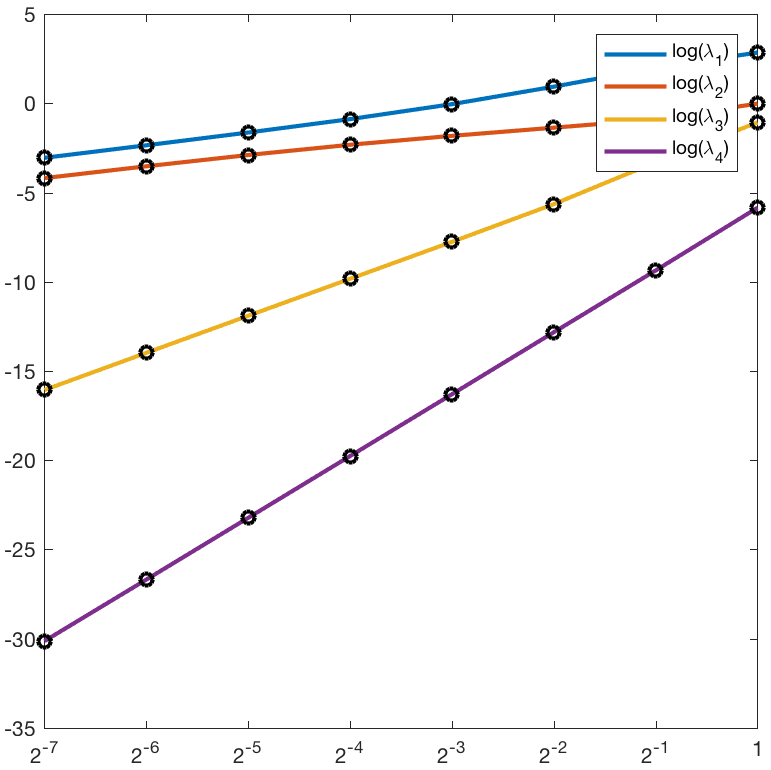}
    \caption{$\log({\rm eig}(P-e^{Fh}Pe^{F'h}))$ vs.\ $h$}
    \label{fig:fig2}
\end{center}
\end{figure}
It is seen that only two eigenvalues remain dominant\footnote{Their ratio tends to the ratio of the non-zero eigenvalues of $GG'$ since $GG'=\lim_{h\to 0}\frac{1}{h}BB'$.} as $h\to 0$. This is in agreement with the spectral density rank being $2$ in the limit.
Dynamical relations between the entries of $\zz$ can be obtained from the limiting continuous-time stochastic model with $(F,G,H)$ as in the second part of Theorem \ref{prop:discretemodel}.}
\end{example}

\section{Subsampling in discrete-time and back}\label{sec:sec5}

In analogy with the correspondence between continuous and
discrete-time stochastic models, we now consider a discrete-time process with the following representation
\begin{subequations}\label{eq:discrete_discrete}
\begin{align}
x_{k+1}&=Fx_k+Gw_k\\
\zz_k&=Hx_k +Jw_k
\end{align}
\end{subequations}
with $w_k$ a discrete-time white noise process and $H\in \mR^{p\times n}$, $F\in \mR^{n\times n}$, $G\in \mR^{n\times m}$, $m\leq p$ as before, with $(F,G)$ controllable and $(F,H)$ observable pairs. Once again
\begin{equation}\label{eq:factorization2}
\Phi_\zz(e^{i\theta})=V_\zz(e^{i\theta})V_\zz(e^{-i\theta})^\prime,
\end{equation}
is the spectral density, with
\[
V_\zz(z)=H(zI-F)^{-1}G+J
\]
a stable spectral factor of $\Phi_\zz$. We use notation $(F,G,H)$, $\Phi$, etc., for this model to reflect the analogy to continuous-time in Section~\ref{sec:c2d}, as it represents the underlying dynamics at the finest time-resolution. For any given positive integer $q$, we consider the subsampled process
\[
\hat \zz_\ell:=\zz_{\ell q}, \mbox{ for }\ell\in\{0,1,2,\ldots\}.
\] 
This admits the stochastic model
\begin{subequations}\label{eq:discrete_discretess}
\begin{align}
\xi_{\ell+1}&=A\xi_\ell+Bv_\ell\\
\hat\zz_\ell&=C\xi_\ell +Dv_\ell
\end{align}
\end{subequations}
where $\xi_\ell := x_{q\ell}$ and 
\begin{subequations}\label{eq:parametersss}
\begin{align}
A&=F^q,\\
B&=\left[\begin{matrix}G,&FG&\ldots,F^{q-1}G\end{matrix}\right],\label{Bhat}\\
C&=H, \mbox{ and}\\
D&=\left[\begin{matrix}0,&0&\ldots,J\end{matrix}\right] \label{Dhat}
\end{align}
\end{subequations}
with
\[
v_\ell=\left[\begin{matrix}w'_{\ell q+q-1},&w'_{\ell q+q-2}&\ldots, &w'_{\ell q}\end{matrix}\right]'.
\]
We note that 
$BB'=\sum_{k=0}^{q-1}F^kGG'(F^k)'
$
gains in rank as $q$ increases and will attain maximal rank for $q\geq n$, and so the spectral density 
\begin{displaymath}
\Psi(z)=W(z)W(z^{-1})' ,
\end{displaymath}
(where $W(z)$ will be specified in \eqref{eq:W}) 
will have rank that is larger than the rank $m$ of $\Phi_\zz(z)$. 
Consequently, as in Section~\ref{sec:c2d}, nullity of $\Phi_\zz(z)$ that is due to dynamic relations between the entries of $\zz$ will no longer be reflected in
$\Psi(z)$.

The analogous statement to Theorem \ref{prop:discretemodel} is the following.

\begin{thm}\label{prop:discretemodelss}
Consider the discrete-time stochastic model \eqref{eq:discrete_discrete} with parameters $(F,G,H,J)$ and $(F,G)$ controllable. 
Then, with $q$ a positive integer, subsampling $\zz_k$ every $q$ steps
gives rise to $\hat\zz_\ell$ in \eqref{eq:discrete_discretess} with parameters $(A,B,C,D)$ satisfying \eqref{eq:parametersss}. Conversely, given a model \eqref{eq:discrete_discretess} with parameters $(A,B,C,D)$, it is 
 a model of a $q$-subsampled discrete-time process $\zeta_k$ provided
\begin{subequations}\label{conditionsdiscr}
\begin{align}
\mbox{i) }&A \mbox{ admits a $q$th root that we denote } A^{1/q}, \mbox{ and}\label{eq:qthroot}\\
\mbox{ii) }& M(P)=\begin{bmatrix}P-A^{1/q}P(A^{1/q})'&A^{1/q-1}BD'\\DB'(A^{1/q-1})'&DD' \end{bmatrix}\geq 0, \nonumber \\
&\mbox{for }P\mbox{ the solution of } P=APA'+BB'. \label{eq:Lambda}
\end{align}
Parameters of a model  \eqref{eq:discrete_discrete} for the corresponding primary process $\zeta_k$ are
$F=A^{1/q}$, $H=C$, and {\footnotesize $\begin{bmatrix}G\\J\end{bmatrix}$} a left factor of $M(P)$ of full column rank, i.e., so that
\begin{equation}\label{eq:GJfactorization}
\begin{bmatrix}G\\J\end{bmatrix}\begin{bmatrix}G'&J'\end{bmatrix} =M(P).
\end{equation}
If $D=0$, \eqref{eq:Lambda} in condition ii) reduces to 
\begin{equation}\label{cond(ii)D=0}
\mbox{ii)$^\prime$}\quad P-A^{1/q}P(A^{1/q})'\geq 0.
\end{equation}
\end{subequations}
\end{thm}

\begin{proof}
The first direction in the theorem has already been shown. It only remains to prove the converse direction. By setting $F=A^{1/q}$ and selecting $G$ and $J$ to satisfy \eqref{eq:GJfactorization}, we see that
$(A,\hat B,C,\hat D)$, with $\hat B$ and $\hat D$ given by \eqref{Bhat} and \eqref{Dhat}, respectively, is a model of a $q$-subsampled discrete process. Now, it remains to show that 
\begin{displaymath}
\begin{bmatrix}\hat{B}\\\hat{D}\end{bmatrix}\begin{bmatrix}\hat{B}'&\hat{D}'\end{bmatrix}
=\begin{bmatrix}B\\D\end{bmatrix}\begin{bmatrix}B'&D'\end{bmatrix}.
\end{displaymath}
To this end observe that 
\[
\hat B\hat B'= P-F^qP(F^q)'=P-APA',
\]
which equals $BB'$ by  condition {\em ii)}. Moreover, $\hat{D}\hat{D}'=JJ'$, which equals $DD'$ by \eqref{eq:GJfactorization}. Finally, by \eqref{eq:GJfactorization} and \eqref{eq:Lambda},
\begin{displaymath}
\hat{B}\hat{D}'= F^{q-1}GJ' =A^{1-1/q}GJ' =BD'.
\end{displaymath}
Hence, $q$-subsampling of the output process of \eqref{eq:discrete_discrete} gives rise to a process with the same statistics as that in \eqref{eq:discrete_discretess}.
\end{proof}


\begin{example}\label{ex:example5}
{\em We now present an academic example where the fastest time-resolution is in discrete time.
We begin with a discrete-time model having spectral density of rank two and no consistent model running with a faster clock being possible. To this end we take
\begin{align*}
x_{k+1}&=\hat F x_k+Gu_k\\
\zz_k&=Hx_k
\end{align*}
with zero constant term $J$, for simplicity, and with $\hat F=e^{Fh}$, $h=0.5$ [time units] and 
$F,G,H$ as in \eqref{eq:FFtotal}.
Assume that data is available at intervals of $h=2.5$ [time units] and, therefore, provide us with the discrete-time model \eqref{eq:discrete_discretess} with parameters
$A=\hat F^5$, $B=[G,\hat F G, \ldots \hat F^4 G]$ and $C=H$.
Selecting $q\in\{2,3,4,5,6,\ldots\}$ and applying the converse direction in Theorem \ref{prop:discretemodelss} gives us models with a faster time rate, as long as $P-A^{1/q}P(A^{1/q})'\geq 0$, with $P$ the solution of $P-A P A '=BB'$. Naturally, for $q=5$ we recover the parameters 
$(\hat F,G,H)$. It can be verified that 
\begin{equation*}
\label{eq:stopping}
P-A^{1/q}P(A^{1/q})'\not\geq 0 \mbox{ for }q>5.
\end{equation*}
Therefore no model with rate faster than $5$ times the observational sampling rate is possible.
 \begin{figure}\begin{center}
 \hspace*{1pt}
   \includegraphics[width=0.38\textwidth]{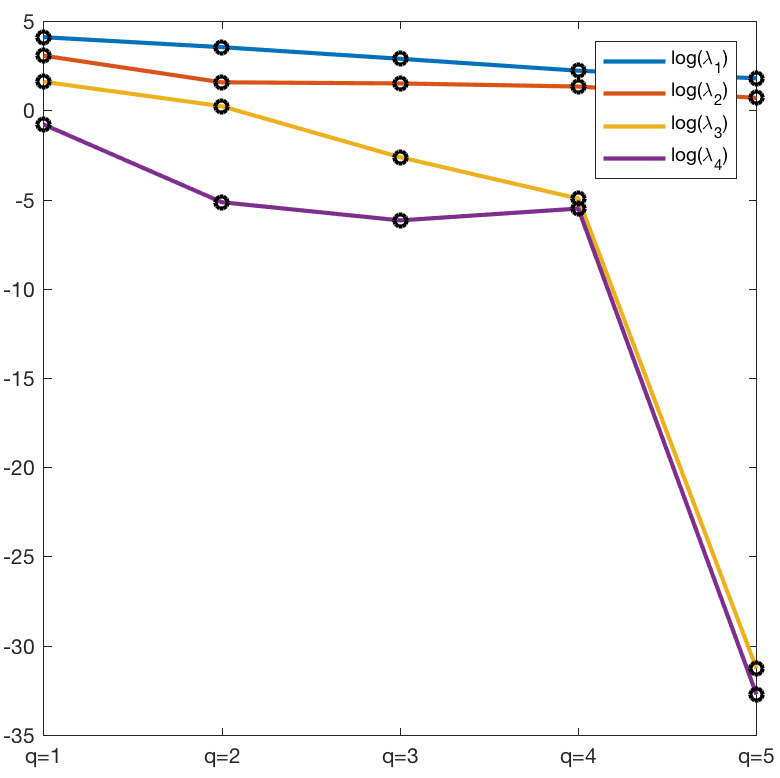}
    \caption{$\log({\rm eig}(P-A^{1/q}P(A')^{1/q}))$ vs.\ $q$}
    \label{fig:fig3}
\end{center}
\end{figure}
As seen in Figure~\ref{fig:fig3}, the lowest two eigenvalues of $P-A^{1/q}P(A^{1/q})'$ separate substantially from the top two and vanish for $q=5$, leading to  $\rank(P-A^{1/q}P(A^{1/q})')=2$. Once again, the parameters $(\hat F,G,H)$ can be obtained from $(A,B,C)$ as in Theorem \ref{prop:discretemodelss}.
}
\end{example}

\section{When is lifting to continuous-time possible?}\label{sec:sec4}

Discrete-time (purely nondeterministic) stationary stochastic models are usually given in the form
\begin{subequations}\label{eq:full_discrete_model}
\begin{align}
\xi_{k+1}&= A \xi_k +Bv_k\\\label{eq:xik}
\zz_k&=C\xi_k + Dv_k 
\end{align}
\end{subequations}
and, accordingly,
\begin{equation}
\label{eq:W}
W(z)=C(zI-A)^{-1}B+D,
\end{equation} 
with a {\em nonzero constant term} $D$.
As we already saw in Example~\ref{ex:example5}, lifting to a continuous-time system may not always be possible, even when $D=0$.
Therefore, we now seek conditions on the parameters of \eqref{eq:full_discrete_model} under which
the output process $\zz_k$ arises via sampling the output of a continuous-time model \eqref{eq:continuous}.
For later reference we define $\bar{C} := CPA'+DB'$ and $\Lambda_0:=CPC'+DD'$ with $P$ satisfying \eqref{eq:discreteLyapunov}. 

Clearly condition \eqref{eq:Psi0} is necessary for
$\zz_k$ to be the sampled output process of \eqref{eq:continuous}. 
So we must have 
\begin{equation}
\label{eq:Psi(0)=0}
\Psi(0)=W(0)W(\infty)'=(D-CA^{-1}B)D'=0.
\end{equation}
There are numerous systems \eqref{eq:full_discrete_model} with the same output $\zz_k$, all having the same $A$ and $C$, but possibly different matrices $B$ and  $D$ \cite[Chapter 6]{LPbook}, each corresponding to a different spectral factor \eqref{eq:W}. Without loss of generality, we choose \eqref{eq:full_discrete_model} to correspond to a minimum phase spectral factor and we let $(B_-,D_-)$ be the corresponding pair of matrices. Then, if the $p\times p$ spectral density $\Psi(e^{i\theta})$ has rank $r$ for all $\theta$, $D_-$ is a $p\times r$ matrix of full column rank  \cite[Theorem 6.1]{LPbook}, and therefore $R:=D_-'D_- >0$. However, \eqref{eq:Psi(0)=0} implies that $(D_--CA^{-1}B_-)R=0$, and therefore
\begin{equation}
\label{eq:Dhat=0}
D_--CA^{-1}B_-=0.
\end{equation}
Next, letting $x_k:=A^{-1}\xi_{k+1}$, 
\begin{displaymath}
x_{k+1}=\xi_{k+1}+A^{-1}B_-v_{k+1}= Ax_k +A^{-1}B_-v_{k+1},
\end{displaymath}\hspace*{-2pt}
and, since $\xi_k\hspace*{-2pt}=\hspace*{-2pt}A^{-1}\xi_{k+1} \hspace*{-2pt}-\hspace*{-2pt}A^{-1}B_-v_k\hspace*{-2pt}=\hspace*{-2pt}x_k \hspace*{-2pt}-\hspace*{-2pt}A^{-1}B_-v_k$, we have $\zeta_k=Cx_k +(D_--CA^{-1}B_-)v_k.$
Consequently, by \eqref{eq:Dhat=0}, 
\begin{align*}
x_{k+1}&= A x_k +A^{-1}B_-v_{k+1}\\
\zz_k&=Cx_k
\end{align*}
which has the form \eqref{eq:discrete} of a sampled system. Now let 
\begin{displaymath}
\hat{P}=A\hat{P}A' +A^{-1}B_-B_-'(A')^{-1}
\end{displaymath}
be the corresponding Lyapunov equation. Then $P_-:=A\hat{P}A'$ satisfies $P_- = AP_-A' +B_-B_-'$, and hence 
\begin{equation}
\label{eq:Pminus}
P_-=\lim_{k\to\infty} \Pi_k
\end{equation}
\cite[Section 6.9]{LPbook}, where $\Pi_{k+1}=\Lambda(\Pi_k)$,  $\Pi_0=0$, with 
\begin{displaymath}
\Lambda(\Pi)=\Pi-A\Pi A' +(\bar{C}-C\Pi A')'\Delta(\Pi)^{-1}(\bar{C}-C\Pi A')
\end{displaymath}
and $\Delta(\Pi):=\Lambda_0 -C\Pi C'$.
We summarize these conclusions in the following theorem.\\[-.05in]

\begin{thm}\label{thm:generaldiscretemodel}
The output process $\zz_k$ of the stochastic system \eqref{eq:full_discrete_model} is the sampled process of the output of a continuous-time system \eqref{eq:continuous} with parameters $(F,G,H)$ provided $\Psi(0)=0$ and conditions \eqref{eq:logm}, \eqref{Bsqinv} and $\log(A)\hat{P}+\hat{P}\log(A)'\leq 0$ hold, where $\hat{P}=A^{-1}P_-(A')^{-1}$ with $P_-$ given by \eqref{eq:Pminus}. Then $\zz_k$ arises by sampling \eqref{eq:continuous} with $F=\frac{1}{h}\log(A)$, $H=C$, and $G$ being a left factor of $-(F\hat{P}+\hat{P}F')$ of full column rank.
\end{thm}


%

\section{Concluding remarks}\label{sec:conclusions}

The most basic underlying principle in modeling is that of parsimony where we seek a minimal number of parameters to explain the data. Historically, statistical reasoning was sought to mathematize the search for exact linear algebraic relations  giving rise to methods of principle component analysis, factor analysis and so on, see e.g., \cite{ning2015linear,picci1986dynamic}. On a parallel route a sizable part of the literature has been devoted to statistical rationales for trading off inaccuracy with model complexity, see e.g., \cite{akaike1998information}. At this point, the art of modeling has accumulated a huge arsenal of ideas that includes regression analysis, likelihood methods, information criteria, regularization, and many more. The importance is further underscored by the mere citation-count of modern accounts of these subjects, e.g., \cite{cohen2013applied} which surpassed hundred thousand citations in few short years.

In the present work we dealt with dynamical models. While the topic has a long history \cite{picci1986dynamic}, issues of how sampling affects dynamical relations have not received sufficient attention. In fact, models are sought at a prespecified time-scale which often coincides with that at which data has been recorded, see e.g., the extensive literature on modern subspace identification methods \cite{van2012subspace}. While sampling rate selection is being discussed in various contexts, in all literature on dynamical relations (e.g., subspace methods, dynamical factor analysis, and so on) the typical assumption is that sampling is ``sufficiently fast'' with the premise that dynamical relations are not impacted if sought at the time-scale of the observation process.

The point of this paper is to bring attention to the fact that, while the sampling rate may or may not be ``fast enough,'' one may retrieve exact dynamical relations at a finer time-scale from models at a coarse time-resolution.
Should the generating mechanism dictate dynamic dependencies that truly originate at time-scale finer than the observational time rate, those dependencies may stay undetected or be poorly identified at the observation rate. 
Building on the paradigm of linear dynamical models and second-order processes we highlight the consequent dictum to seek models at the finest possible time-resolution that are consistent with the data at the given observation time scale.
The implication is that, testing the dimenisionality of ``lifted models''
is a logical indicator of parsimony in system identification.

In summary, our setting involves dynamic (i.e., difference differential) linear relations between variables and stochastic noise. Guidelines that emerge can be summarized as follows:
\begin{itemize}
\item[a)] Given data at observation time-intervals, determine linear dynamical models with parameters $(A,B,C,D)$ using mostly standard techniques.
\item[b)] Determine whether the corresponding discrete time model originates from a continuous-time one via sampling the output process, and if not, determine if there is a discrete-time model at a faster time-scale which is consistent.
\item[c)] Dynamic relations between variables ought to be sought for the {\em lifted model at the finest time-scale}.
\end{itemize}

System identification techniques may be tuned to the idea that parsimony is sought at a scale other than that at which the system $(A,B,C,D)$ is seen to operate. To this end, low-rank regularizers and the postulate of observational noise could be redirected into fitting low complexity models of the {\em lifted } model $(F,G,H)$ of e.g., Theorem \ref{prop:discretemodel}. Thus, the principal contribution of this work is to introduce the idea of {\em lifting identified models to a finer time-scale before assessing their complexity}.
Sampling or subsampling a random process, and generating data at a rate that is lower than the native time scale, masks the exactness of any pre-existing dynamic relations.

\bibliographystyle{IEEEtran}
\bibliography{refs}

\end{document}